\documentclass[11pt]{article} 
\usepackage{fullpage}
\usepackage[latin1]{inputenc} 

\usepackage{amsmath}
\usepackage{amsthm}
\usepackage{amssymb}
\usepackage{graphicx}
\usepackage{pstricks}
\usepackage{hyperref}
\usepackage{authblk}
\usepackage{ifpdf}

\ifpdf
\fi

\newtheorem{theorem}{Theorem}[section]
\newtheorem{conj}[theorem]{Conjecture}

\newtheorem{lemma}[theorem]{Lemma}
\newtheorem{fact}[theorem]{Fact}
\newtheorem{claim}[theorem]{Claim}
\newtheorem{example}[theorem]{Example}
\newtheorem{remark}[theorem]{Remark}
\theoremstyle{definition}
\newtheorem{defn}[theorem]{Definition}

\DeclareMathOperator{\tw}{\sf{tw}}
\DeclareMathOperator{\pw}{\sf{pw}}
\DeclareMathOperator{\layout}{\sf{Layout}}
\DeclareMathOperator{\MCLA}{\sf{MCLA}}
\DeclareMathOperator{\MLA}{\sf{MLA}}

\newcommand{\osbp}{{\sf BP}^{1{\sf s}}}
\newcommand{\osbwp}{{\sf BWP}^{1{\sf s}}}

\DeclareMathOperator{\SSE}{\sf{SSE}}
\DeclareMathOperator*{\agg}{agg}
\newcommand{\reals}{\mathbb{R}}

\newcommand{\comment}[2]{}

\title{Inapproximability of Treewidth,\\ One-Shot Pebbling, and Related Layout Problems\thanks{Research supported by NSERC.}}
\author{Per Austrin}
\author{Toniann Pitassi}
\author{Yu Wu}
\affil{
  Department of Computer Science\\
  University of Toronto\\
  \texttt{\{austrin,toni,wuyu\}@cs.toronto.edu}
}

\begin{document}

\begin{titlepage}
  \maketitle
  \thispagestyle{empty}
  
  \begin{abstract}
    We study the approximability of a number of graph problems:
    treewidth and pathwidth of graphs, one-shot black (and black-white)
    pebbling costs of directed acyclic graphs, and a variety of
    different graph layout problems such as minimum cut linear
    arrangement and interval graph completion.  We show that, assuming
    the recently introduced Small Set Expansion Conjecture, all of these
    problems are hard to approximate within any constant factor.
  \end{abstract}
\end{titlepage}
  
%\tableofcontents

\section{Introduction}

One of the great accomplishments in the last twenty years in
complexity theory has been the development of ideas that has led to a
deep understanding of the approximability of an astonishing number of
NP-hard optimization problems.  More recently, in the last ten years,
the formulation of the Unique Games Conjecture (UGC) due to Khot
\cite{Khot02} has inspired a remarkable body of work, clarifying the
complexity of many optimization problems, and exposing the central
role of semidefinite programming in the development of approximation
algorithms.

Despite this tremendous progress, for certain expansion problems such
as the $c$-Balanced Separator problem, and graph layout problems such
as the Minimum Linear Arrangement (MLA) problem, their approximation
status remained unresolved.  That is, even assuming the UGC is not
known to be sufficient to obtain hardness of approximation for either
of these problems.  Moreover, the approximability of many other graph
layout problems is similarly unresolved, even under the UGC.
Intuitively this is because the hard instances for these problems seem
to require a certain global structure such as expansion. Typical
reductions for these problems are gadget reductions which preserve
global properties of the unique games instance, such as the lack of
expansion. Therefore, barring radically new types of reductions that
do not preserve global properties, proving hardness for $c$-Balanced
Separator seems to require a stronger version of UGC, where the
instance is guaranteed to have good expansion.

In \cite{RS10}, the Small Set Expansion (SSE) Conjecture was
introduced, and it was shown that it implies the UGC, and that the SSE
Conjecture follows if one assumes that the UGC is true for somewhat
expanding graphs.  In follow-up work by Raghavendra et
al.\ \cite{RST10}, it was shown that the SSE Conjecture is in fact
equivalent to the UGC on somewhat expanding graphs, and that the SSE
Conjecture implies hardness of approximation for $c$-Balanced
Separator and MLA.  In this light, the Small Set Expansion conjecture
serves as a natural unified conjecture that yields all of the
implications of UGC and also hardness for expansion-like problems that
appear to be beyond the reach of the UGC.

In this paper, we study the approximability of a host of such graph
layout problems, including: treewidth and pathwidth of graphs,
one-shot black and black-white pebbling, Minimum Cut Linear
Arrangement (MCLA) and Interval Graph Completion (IGC).
We prove that all of these problems are SSE-hard to
approximate to within any constant factor.
Our main contributions, giving SSE-hardness of approximation for all of
the graph layout problems mentioned above, are described in the following
subsections. For all of these problems, no evidence of hardness of
approximation was known prior to our results.

It should be noted that the status of the SSE Conjecture is very open
at this point.  In particular, by the recent result of Arora et
al.\ \cite{ABS10} (see also subsequent work \cite{BRS11,GS11}), it has
algorithms running in subexponential time.  Still, despite this recent
progress providing negative evidence against the SSE Conjecture, it
remains open, and we think that investigating what open problems in
approximability we can show SSE-hardness for is a worthwhile venture.

\subsection{Width Parameters of Graphs}

The \emph{treewidth} of a graph, introduced by Robertson and Seymour
\cite{rs84,rs86}, is a fundamental parameter of a graph that measures
how close a graph is to being a tree.  The concept is very important
since problems of small treewidth can usually be solved efficiently by
dynamic programming. Indeed, a large body of NP-hard problems
(including all problems definable in monadic second-order logic
\cite{Courcelle90}) are solvable in polynomial time and often even
linear time on graphs of bounded treewidth.  Examples of such
optimization problems include finding the maximum independent set in a
graph, as well as finding Hamiltonian cycles.  In machine learning,
tree decompositions play a key role in the development of efficient
algorithms for fundamental problems such as probabilistic inference,
constraint satisfaction and query optimization. (See the excellent
survey \cite{Bod05} for motivation, including theoretical as well as
practical applications of treewidth.)

The complexity of approximating treewidth is a longstanding open
problem.  Determining the exact treewidth of a graph and producing an
associated optimal tree decomposition (see Definition~\ref{def:tw}) is
known to be NP-hard \cite{ACP87}.  A central open problem is to
determine whether or not there exists a polynomial time constant
factor approximation algorithm for treewidth
(see e.g., \cite{BGHK95,feige,Bod05}).  The current best polynomial time
approximation algortihm for treewidth \cite{feige}, computes the
treewidth $\tw(G)$ within a factor $O(\sqrt{\log \tw(G)})$.  On the
other hand, the only hardness result to date for treewidth shows that
it is NP-hard to compute treewidth within an {\it additive} error of
$n^{\epsilon}$ for some $\epsilon >0$ \cite{BGHK95}.  No hardness of
approximation is known and not even the possibility of a
polynomial-time approximation scheme for treewidth has been ruled out.
In many important special classes of graphs, such as planar graphs
\cite{ST94}, asteroidal triple-free graphs \cite{BT03}, and
$H$-minor-free graphs \cite{feige}, constant factor approximations are
known, but the general case has remained elusive.

On the positive side, there is a large body of literature developing
fixed-parameter algorithms for treewidth.  In particular, when the
runtime is allowed to be exponential in the $\tw(G)$ there are
constant factor approximations.  Furthermore, even exactly determining
the treewidth is fixed-parameter tractable: there is a linear time
algorithm for computing the (exact) treewidth for graphs of constant
treewidth \cite{Bod96}.

A related graph parameter is the so-called \emph{pathwidth}, which can
be viewed as measuring how close $G$ is to a path.  The pathwidth
$\pw(G)$ is always at least $\tw(G)$, but can be much larger.  The
current state of affairs here is similar as for treewidth; though the
current best approximation algorithm only has an approximation ratio
of $O(\sqrt{\log \pw(G)} \log n)$ \cite{feige}, the best hardness
result is NP-hardness of additive $n^\epsilon$ error approximation.

Using the recently proposed \emph{Small Set Expansion} (SSE)
Conjecture \cite{RS10} discussed earlier, we show that both $\tw(G)$
and $\pw(G)$ are hard to approximate within any constant factor.  In
fact, we show something stronger: it is hard to distinguish graphs
with small pathwidth from graphs with large treewidth.  Specifically:

\begin{theorem}\label{thm:treewidth}
  For every $\alpha > 1$ there is a $c > 0$ such that given a graph $G =
  (V,E)$ it is SSE-hard to distinguish between the case when $\pw(G)
  \le c \cdot |V|$ and the case when $\tw(G) \ge \alpha \cdot c \cdot |V|$.

  In particular, both treewidth and pathwidth are SSE-hard to
  approximate within any constant factor.
\end{theorem}

This is the first result giving hardness of (relative) approximation
for these problems, and gives evidence that no constant factor
approximation algorithm exists for either of them.

\subsection{Pebbling Problems}

Graph pebbling is a rich and relatively mature topic in theoretical
computer science.  \emph{Pebbling} is a game defined on a directed
acyclic graph (DAG), where
the goal is to \emph{pebble} the sink nodes of the DAG according to
certain rules, using the minimum number of pebbles.  The rules for
pebbling are as follows.  A \emph{black pebble} can be placed on a
node if all of the node's immediate predecessors contain pebbles, and
can always be removed.  A white pebble can always be placed on a node,
but can only be removed if all of the node's immediate predecessors
contain pebbles.  A pebbling {\it strategy} is a process of pebbling
the sink nodes in a graph according to the above rules.  The
\emph{pebbling cost} of a pebbling strategy is the maximum number of
pebbles used in the strategy.  The black-white pebbling cost of a DAG
is the minimum pebbling cost of all possible pebbling strategies. The
black pebbling cost is the minimum pebbling cost over all pebbling
strategies that only use black pebbles.

Pebbling games were originally devised for studying
programming languages and compiler construction, but have later found
a broad range of applications in computational complexity theory.
Pebbling is a tool for studying the relationship between computation
time and space by means of a game played on directed acyclic graphs.
It was employed to model register allocation, and to
analyze the relative power of time and space as Turing machine
resources.  For a comprehensive recent survey on graph pebbling, see
\cite{Nor10}.

Apart from the cost of a pebbling, another important measure is the
\emph{pebbling time}, which is the number of steps (pebble
placements/removals) performed.  In the context of measuring memory
used by computations, this corresponds to computation time, and hence
keeping the pebbling time small is a natural priority.  The extreme
case of this is what we refer to as \emph{one-shot pebbling}, also
known as progressive pebbling, considered in
e.g.\ \cite{sethi,Lengauer81,KP86}.  In one-shot pebbling, we have the
restriction that each node can receive a pebble only once.  Note that
this restriction can cause a huge increase in the pebbling cost of the
graph \cite{LT82}.  One-shot pebbling is also equivalent to a problem
known as Register Sufficiency \cite{ravi}.

The one-shot pebbling problem is easier to analyze for the following
reasons.  In the original pebbling problem, in order to achieve the
minimum pebbling number, the pebbling time might be required to be
exponentially long, which becomes impractical when $n$ is large.  On
the other hand, the one-shot pebbling problem is more amenable to
complexity theoretic analysis as it minimizes the space used in a
computation subject to the execution time being minimum.  In
particular, the decision problem for one-shot pebbling is in NP
(whereas the unrestricted pebbling problems are PSPACE-complete).

The one-shot black pebbling problem and one-shot black-white pebbling
problems admit an $O(\sqrt{\log n} \log n)$ approximation ratio.  We
show that they are SSE-hard to approximate to within any constant
factor.  For black pebbling we show that this holds for single sink
DAGs with in-degree $2$, which is the canonical setting for pebbling
games (it seems plausible that the black-white hardness can be shown
to hold for this case as well, though we have not attempted to prove
this).

\begin{theorem} \label{thm:osbp}
  It is SSE-hard to approximate the one-shot black pebbling problem
  within any constant factor, even in DAGs with a single sink and
  maximum in-degree $2$.
\end{theorem}

\begin{theorem} \label{thm:osbwp}
  It is SSE-hard to approximate the one-shot black-white pebbling problem within any constant factor.
\end{theorem}

No hardness of approximation result of any form was known for one-shot
pebbling problems.  We believe that these results can be extended to
obtain hardness for more relaxed versions of bounded time pebbling
costs as well.  We are currently working on this, and have some
preliminary results.

\subsection{The Connection: Layout Problems}
\label{sec:intro:layout}

The graph width and one-shot pebbling problems discussed in
the previous sections may at first glance appear to be unrelated.  However,
both sets of problems are instances of a general family of
problems, known as \emph{graph layout problems}.  In a graph layout
problem (also known as an arrangement problem, or a vertex ordering
problem), the goal is to find an ordering of the vertices, optimizing
some condition on the edges, such as adjacent pairs being close.
Layout problems are an important class of problems that have
applications in many areas such as VLSI circuit design.

A classic example is the \emph{Minimum Cut Linear Arrangement} Problem
(MCLA).  In this problem, the objective is to find a permutation $\pi$
of the vertices $V$ of an undirected graph $G = (V, E)$, such that the
largest number of edges crossing any point,
\begin{equation}
  \label{eq:mcla_intro}
  \max_i |\{ (u,v) \in E | \pi(u) \le i < \pi(v) \}|,
\end{equation}
is minimized.  MCLA is closely related to the \emph{Minimum Linear
  Arrangement} Problem (MLA), in which the $\max$ in
\eqref{eq:mcla_intro} is replaced by a sum.

The MCLA problem can be approximated to within a factor $O(\log n
\sqrt{\log n})$.  To the best of our knowledge, there is no
hardness of approximation for MCLA in the literature.  Its cousin MLA was
recently proved SSE-hard to approximate within any constant factor
\cite{RST10}, and we observe that the same hardness applies to the
MCLA problem.

\begin{theorem} 
  \label{thm:mcla}
  Assuming the SSE Conjecture, Minimum Cut Linear Arrangement is hard to approximate within any constant factor.
\end{theorem}

Another example of graph layout is the \emph{Interval Graph
  Completion} Problem (IGC).  In this problem, the objective is to
find a supergraph $G' = (V, E')$ of $G$ such that $G'$ is an interval
graph (i.e., the intersection graph of a set of intervals on the real
line) and of minimum size.  While not immediately appearing to be a
layout problem, using a simple structural characterization of interval
graphs \cite{RR88} one can show that IGC can be reformulated as
finding a permutation of the vertices that minimizes the sum over the
longest edges going out from each vertex, i.e., minimizing
\begin{equation}
  \label{eq:igc_intro}
  \sum_{u \in V} \max_{(u,v) \in E} \max \{\pi(v) - \pi(u), 0\}.
\end{equation}
See e.g., \cite{CHKR10}.  The current best approximation algorithm for
IGC achieves a ratio of $O(\sqrt{\log n} \log \log n)$ \cite{CHKR10}.
It turns out that the SSE Conjecture can be used to prove
super-constant hardness for this problem as well.

\begin{theorem} 
  \label{thm:igc}
  Assuming the SSE Conjecture, Interval Graph Completion is hard to approximate within any constant factor.
\end{theorem}

There is a distinction in IGC of whether one counts the number of
edges in the final interval graph -- this is the most common
definition -- or whether one only counts the number of edges added to
make $G$ an interval graph (which makes the problem harder from an
approximability viewpoint).  Our result holds for the common
definition and therefore applies also to the harder version.

Theorems~\ref{thm:mcla} and \ref{thm:igc} are just two examples of
layout problems that we prove hardness of approximation for.
By varying the precise objective function and also considering
directed acyclic graphs, in which case the permutation $\pi$ must be a
topological ordering of the graph, one can obtain a wide variety of
graph layout problems.  We consider a set of eight such problems,
generated by three natural variations (see
Section~\ref{sec:layout-defs} for precise details), and show
super-constant SSE-based hardness for all of them in a unified way.
This set of problems includes MLA, MCLA, and IGC, but not problems
such as Bandwidth (but on the other hand, strong NP-hardness
inapproximability results for Bandwidth are already known \cite{DFU11}).
See Table~\ref{table:taxonomy} in Section~\ref{sec:layout-defs} for a
complete list of problems covered.

\begin{theorem}
  \label{thm:layout}
  Assuming the SSE Conjecture, all problems listed in
  Table~\ref{table:taxonomy} (see page~\pageref{table:taxonomy}) are
  hard to approximate to within any constant factor.
\end{theorem}

Let us now return to the problems discussed in the previous sections.
It should not be surprising that the one-shot black pebbling problem
is equivalent to a graph layout problem: the one-shot constraint
reduces the problem to determining in which order to pebble the
vertices; such an ordering induces a pebbling strategy in an obvious
way.  For the black-white case, it is known that the one-shot
black-white pebbling cost of $D$ is interreducible with a layout
problem on an undirected graph $G$.  Both of these layout problems are
included in the set of problems we show hardness for, so
Theorems~\ref{thm:osbp} and \ref{thm:osbwp} follow immediately from
Theorem~\ref{thm:layout}.

Turning to the width parameters, treewidth is equivalent to a graph
layout problem called elimination width.  Here the objective function
is somewhat more intricate than in the set of basic layout problems we
consider in Theorem~\ref{thm:layout}, but we are able to extend those
results to hold also for elimination width.  Pathwidth is also known
to be equivalent to a certain graph layout problem, and in fact is
equivalent to the layout problem which one-shot black-white pebbling
reduces to.  We use these connections to prove the hardness of
approximation for both treewidth and pathwidth, thereby obtaining
Theorem~\ref{thm:treewidth}.

\subsection{Previous Work}

As the reader may have noticed, for all the problems mentioned, the
best current algorithms achieve similar poly-logarithmic approximation
ratios.  Given their close relation, this is of course not surprising.
Most of the algorithms are obtained by recursively applying some
algorithm for the $c$-balanced separator problem, in which the
objective is to find a bipartition of the vertices of a graph such
that both sides contain a $c$ fraction of vertices, and the number of
edges crossing the partition is minimized.

In the pioneering work on separators by Leighton and Rao
\cite{Leighton}, an $O(\log n)$ approximation algorithm for
$c$-balanced separator was given, which was used to design $O(\log^2
n)$ approximation algorithm for a number of graph layout problems such
as MLA, MCLA, and Register Sufficiency.  Later, \cite{RR98} improved
the approximation algorithm for MLA to a ratio $O(\log n \log \log
n)$, using a spreading metric method.  In the groundbreaking work of
Arora et al.~\cite{arv}, semidefinite programming was used to give an
improved approximation ratio of $O(\sqrt{\log n})$ for $c$-balanced
separator.  Using their ideas, improved algorithms for ordering
problems have been found, such as the $O(\sqrt{\log n} \log \log n)$
approximation algorithm for IGC and MLA \cite{CHKR10}, the
$O(\sqrt{\log n})$ approximation algorithm for treewidth \cite{feige}
and the $O(\sqrt{\log n} \log n)$ approximation algorithm for
pathwidth \cite{feige}.

It is known that the register sufficiency problem (also known as
one-shot black pebbling) admits a $O(\log^2 n)$ approximation
algorithm \cite{ravi}.  We observe that by plugging in the improved
approximation algorithm for direct vertex separator
\cite{directed-edge} into the algorithm in \cite{ravi}, one can
improve this to an $O(\sqrt{\log n}\log n)$ approximation algorithm.

Again, in these algorithms, the approximation algorithm for
$c$-balanced separator plays a key role.  An improved algorithm for
$c$-balanced separator will also improve the approximation algorithms
for the other problems.  On the other hand, hardness of approximating
$c$-balanced separator \cite{RST10} does not necessarily imply
hardness of approximating layout problems.

On the hardness side, our work builds upon the work of \cite{RST10},
which showed that the SSE Conjecture implies superconstant hardness of
approximation for MLA (and for $c$-balanced separator).  The only
other hardness of relative approximation that we are aware of for
these problems is a result of Amb\"uhl et al.~\cite{AMS07}, showing
that MLA does not have a PTAS unless NP has randomized subexponential
time algorithms.

\subsection{Organization}

The outline for the rest of the paper is as follows.  In
Section~\ref{sec:defs}, we formally define the layout problems studied
as well as treewidth and pathwidth.  Section~\ref{sec:overview} gives
an overview of the reductions used.  Then Section~\ref{sec:layout}
gives the reductions proving Theorem~\ref{thm:layout},
Section~\ref{sec:treewidth} we use that to prove
Theorem~\ref{thm:treewidth}, and in Section~\ref{sec:pebbling
  postprocess} we give some additional reductions for our pebbling
instances in order to achieve indegree $2$ and single sinks, as
promised in Theorem~\ref{thm:osbp}.  Finally we end with some
concluding remarks and open problems in Section~\ref{sec:conclusion}.

\section{Definitions and Preliminaries} \label{sec:defs}

\subsection{Graph Layout Problems}
\label{sec:layout-defs}

In this section, we describe the set of graph layout problems that we
consider.  A problem from the set is described by three parameters,
giving rise to several different problems.  These three parameters are
by no means the only interesting parameters to consider (and some of
the settings give rise to more or less uninteresting layout problems).
However, they are sufficient to capture the problems we are interested
in except treewidth, which in principle could be incorporated as
well though we refrain from doing so in order to keep the definitions
simple (see Section~\ref{sec:treewidth-defs} for more details).

First a word on notation.  Throughout the paper, $G = (V, E)$ denotes
an undirected graph, and $D = (V,E)$ denotes a directed (acyclic)
graph.  Letting $n$ denote the number of vertices of the graph, we are
interested in bijective mappings $\pi: V \rightarrow [n]$.  We say
that an edge $(u,v) \in E$ \emph{crosses} point $i \in [n]$ (with
respect to the permutation $\pi$, which will always be clear from
context), if $\pi(u) \le i < \pi(v)$.

We consider the following variations:
\begin{enumerate}

\item \textbf{Undirected or directed acyclic}: In the case of an
  undirected graph $G$, any ordering $\pi$ of the vertices is a
  feasible solution.  In the case of a DAG $D$, only the
  topological orderings of $D$ are feasible solutions.

\item \textbf{Counting edges or vertices}: for a point $i \in [n]$ of the
  ordering, we are interested in the set $E_i(\pi)$ of edges crossing this
  point.  When counting edges, we use the cardinality of $E_i$ as our
  basic measure.  When counting vertices, we only count the set of
  vertices $V_i$ to the left of $i$ that are incident upon some edge
  crossing $i$.  In other words, $V_i$ is the projection of $E_i(\pi)$ to
  the left-hand side vertices.  Formally:
  \begin{align*}
    E_i(\pi) &= \{ e \in E \,|\, \pi(u) \le i < \pi(v) \text{ where $e = (u,v)$}\} \\
    V_i(\pi) &= \{ u \in V \,|\, \pi(u) \le i < \pi(v) \text{ for some $(u,v) \in E$}\}
  \end{align*}
  We refer to $|E_i(\pi)|$ or $|V_i(\pi)|$ (depending on whether we are counting
  edges or vertices) as the \emph{cost} of $\pi$ at $i$.

\item \textbf{Aggregation by sum or max}: given an ordering $\pi$, we
  aggregate the costs of each point $i \in [n]$, by either summation
  or by taking the maximum cost.
\end{enumerate}

Given these choices, the objective is to find a feasible ordering
$\pi$ that minimizes the aggregated cost.

\begin{defn}(Layout value)
  For a graph $H$ (either an undirected graph $G$ or a DAG $D$), a
  cost function $C$ (either $E$ or $V$), and an aggregation function
  $\agg: \reals^* \rightarrow \reals$ (either $\Sigma$ or $\max$), we define
  $\layout(H; C, \agg)$ as the minimum aggregated cost over all feasible orderings of $H$.  Formally:
  $$
  \layout(H; C, \agg) = \min_{\textrm{feasible $\pi$}} \, \agg_{i \in [n]} |C_i(\pi)|.
  $$
\end{defn}

\begin{example}
  $$\layout(G; E, \max) = \min_{\pi}
  \max_{i \in [n]} |E_i(\pi)|,$$ where $\pi$ ranges over all orderings
  of $V(G)$.  This we recognize from Section~\ref{sec:intro:layout} as
  the Minimum Cut Linear Arrangement value of $G$.  
\end{example}

\begin{example}
  $$\layout(D; V, \max) = \min_{\pi} \max_{i \in [n]}
  |V_i(\pi)|,$$ where $\pi$ ranges over all topological orderings of
  the DAG $D$.  As we shall see in Section~\ref{sec:pebbling-defs}, this is
  precisely the One-Shot Black Pebbling cost of $D$.
\end{example}

Combining the different choices gives rise to a total of eight layout
problems (some more natural than others).  Several of these appear in
the literature under one or more names, and some turn out to be
equivalent\footnote{Here, we consider two optimization problems equivalent if there are reductions between them that change the objective values by at most an additive constant.  } to problems that at first sight appear to be different.  We
summarize some of these names in Table~\ref{table:taxonomy}.  In some
cases the standard definitions of these problems look somewhat
different than the definition given here (e.g., for pathwidth,
one-shot pebblings, and interval graph completion).  For the
pebbling and pathwidth problems, we discuss these equivalences
of definitions in the following two sections.

\begin{table}
  \centering
  \begin{tabular}{|c|c|c||p{9cm}|}
    \hline
    \multicolumn{3}{|c||}{\textbf{Problem}} & \multicolumn{1}{|c|}{\textbf{Also known as / Equivalent with }} \\
    \hline
    \hline
    undir. & edge & sum & Minimum/Optimal Linear Arrangement  \\
    \hline
    undir. & edge & max & Minimum Cut Linear Arrangement \newline CutWidth \\
    \hline
    undir. & vertex & sum & Interval Graph Completion \newline SumCut \\
    \hline
    undir. & vertex & max & Pathwidth \newline One-shot Black-White Pebbling \newline Vertex Separation \\
    \hline
    DAG & edge & sum & Minimum Storage-Time Sequencing\newline Directed MLA/OLA   \\
    \hline
    DAG & edge & max &       \\
    \hline
    DAG & vertex & sum &      \\
    \hline
    DAG & vertex & max &  One-shot Black Pebbling  \newline Register Sufficiency\\
    \hline
  \end{tabular}
  \caption{Taxonomy of Layout Problems}
  \label{table:taxonomy}
\end{table}

For interval graph completion, recall from
Section~\ref{sec:intro:layout} that the objective is to minimize
\begin{equation*}
  \sum_{u \in V} \max_{(u,v) \in E} \max \{\pi(v) - \pi(u), 0\}.
\end{equation*}
In other words, we are counting the longest edge going to the right
from each point $i$.  If the length of this edge is $l$ then the edge
contributes $1$ to $V_{i}(\pi), \ldots, V_{i+l-1}(\pi)$ and hence the
objective can be rewritten as
\begin{equation*}
  \sum_{u \in V} |V_i(\pi)|,
\end{equation*}
so that Interval Graph Completion is precisely $\layout(G; V, \Sigma)$.

\subsection{Treewidth, Elimination Width, and Pathwidth}
\label{sec:treewidth-defs}

\begin{defn}[Tree decomposition, Treewidth] \label{def:tw}
Let $G=(V,E)$ be a graph, $T$ a tree, and let $\mathcal{V} = (V_t)_{t\in T}$ be a family of vertex sets $V_t \subseteq V$ indexed by the vertices $t$ of $T$. The pair $(T,\mathcal{V})$ is called a \emph{tree decomposition} of $G$ if it satisfies the following three conditions:
\begin{enumerate}
 \item[(T1)] $V= \cup_{t\in T} V_t$;
\item[(T2)] for every edge $e \in E$, there exists a $t\in T$ such that both endpoints of $e$ lie in $V_t$;
\item[(T3)] for every vertex $v \in V$, $\{t \in T \,|\,v \in V_t\}$ is a subtree of $T$'.
\end{enumerate}

The width of $(T,\mathcal{V})$ is the number $ \max \{|V_t|-1 \,|\, t\in T\},$
and the \emph{treewidth} of $G$, denoted $\tw(G)$,  is the minimum width of any tree decomposition of $G$.
\end{defn}

\begin{defn}
Let $G = (V,E)$ be a graph, and let $v_1, \ldots, v_n$ be some
ordering of its vertices.  Consider the following process: for each
vertex $v_i$ in order, add edges to turn the neighborhood of $v_i$
into a clique, and then remove $v_i$ from $G$.  This is an
\emph{elimination ordering} of $G$.  The \emph{width} of an
elimination ordering is the maximum over all $v_i$ of the degree of
$v_i$ when $v_i$ is eliminated.  The \emph{elimination width} of $G$
is the minimum width of any elimination order.
\end{defn}

\begin{theorem}[See e.g., \cite{Bod07}]
  For every graph $G$, the elimination width of $G$ equals $\tw(G)$.
\end{theorem}

Thus treewidth is another example of a layout problem.  In principle
this layout problem can be formulated in the framework of
Section~\ref{sec:layout-defs}, but the choice of cost function is now
more involved than the vertex- and edge-counting considered there.

\begin{defn}[Path decomposition, Pathwidth] \label{pw}
Given a graph $G$, we say that $(T,\mathcal{V})$ is a \emph{path
  decomposition} of $G$ if it is a tree decomposition of $G$ and $T$
is a path.  The \emph{pathwidth} of $G$, denoted $\pw(G)$,  is the minimum width of any path decomposition of $G$.
\end{defn}

As claimed earlier, pathwidth is in fact equivalent with a graph
layout problem:

\begin{theorem}[\cite{nancy}]
  For every graph $G$, we have $\pw(G) = \layout(G; V, \max)$, also
  known (among many other names) as the ``vertex separation'' number
  of $G$.
\end{theorem}

\subsection{Pebbling Problems}
\label{sec:pebbling-defs}

In this section we define pebbling problems and their one-shot versions.

\begin{defn} (Pebbling Configurations)
  Let $D=(V,E)$ be a directed acyclic graph (DAG).  A \emph{pebbling
    configuration} of $D$ is a pair $(B,W)$ of (disjoint) subsets vertices
  (representing the set $B$ of vertices that have black pebbles, and
  the set $W$ of vertices that have white pebbles on them).
\end{defn}

\begin{defn} (Black and Black-White Pebbling Strategies) \label{defn2}
Let $D=(V,E)$ be a directed acyclic graph.
A \emph{black-white pebbling strategy} for $D$ is a sequence of pebble
configurations $\mathcal{P} = \{P_0, \ldots, P_{\tau} \}$ such that:
\begin{enumerate}
\item[(i)] the first and last
configurations contain no pebbles; that is $P_0 = P_{\tau} = (\emptyset,\emptyset)$. 
\item[(ii)] each sink vertex $u$ of $D$ is pebbled at least once, i.e., there is some $P_t = (B_t, W_t)$ such that $u \in B_t \cup W_t$.
\item[(iii)] each configuration follows
  from the previous configuration by one of the following rules:
  \begin{enumerate}
  \item A black pebble can be removed from a vertex.
  \item A black pebble can be placed on a pebble-free vertex $v$ if all of the immediate predecessors of $v$ are pebbled.
  \item A white pebble can be placed on a pebble-free vertex.
  \item A white pebble can be removed from a vertex $v$ if all of the immediate predecessors of $v$ are pebbled.
  \end{enumerate}
\end{enumerate}
A black pebbling strategy for $G$ is a black-white pebbling strategy in which no white pebbles are used.
\end{defn}

The \emph{cost} of a pebbling strategy is
$cost(\mathcal{P})=\max_{0\leq t\leq \tau} \{|B_t \cup W_t|\}$.  The
black-white pebbling cost of $D$ is the minimum cost of any
black-white pebbling strategy of $D$, and similarly the black pebbling
cost of $D$ is the minimum cost of any black pebbling strategy of $D$.

\begin{defn} \label{defn4} (One-Shot Black and One-Shot Black-White Pebbling)
  A \emph{one-shot black (resp.\ black-white) pebbling strategy} is a
  black (resp.\ black-white) pebbling strategy in which each node is
  only pebbled once.  The one-shot black (resp.\ black-white) pebbling
  cost of $D$, denoted $\osbp(D)$ (resp.\ $\osbwp(D)$) is the minimum
  cost of any one-shot black (resp.\ black-white) pebbling strategy of
  $D$.
\end{defn}

As mentioned in Table~\ref{table:taxonomy}, the one-shot pebbling
problems can be formulated as $\layout$ problems.

\begin{lemma} 
For every DAG $D=(V,E)$, we have $\osbp(D) = \layout(D,V,\max)$.
\end{lemma}

\begin{proof}
Suppose $\pi$ is the optimal ordering of $\layout(D,V,\max)$, we pebble the vertices according to $\pi$.  We remove a pebble from vertex $u$ if and only if all of the successors of $u$ are pebbled. Since $\pi$ is a topological order of $D$, this is a valid pebbling strategy. 
It is easy to verify that after pebbling $\pi(i)$, the number of pebbles on the graph is $|V_i(\pi)|$. Therefore the number of pebbles used in the above strategy is $\layout(D,V,\max)$. 
On the other hand, suppose $\Gamma$ is the optimal pebbling strategy, let $\sigma$ be the  ordering of vertices to receive a pebble in $\Gamma$. We consider the number of pebbles on the graph after pebbling the $i$-th vertex in $\sigma$. For any vertex $u$ that has a pebble, if the vertex has a successor that has not yet be pebbled, then the pebble on $u$ cannot be removed, since $u$ cannot be pebbled again. Therefore the number of pebbles on the graph is at least $V_i(\sigma)$. 
Thus $\osbp(D) \ge \max_{i \in [n]} |V_i(\sigma)| \geq \layout(D,V,\max) $. 
\end{proof}

For one-shot black-white pebbling, we have the following reductions by
Lengauer \cite{Lengauer81}, showing that one-shot black-white pebbling
is equivalent to the undirected max-vertex layout problem.

\begin{lemma}[\cite{Lengauer81}]
  For a given DAG $D=(V,E)$, let $G_D = (V, E_D)$ be an undirected
  graph with $E_D = \{(v,w) \,|\, (v,w)\in E\} \cup \{ (v,w) \,|\,
  \exists u, (v,u),(w,u) \in E\}$.  Then $$\osbwp(D) = \layout(G_D, V,
  \max) - 1.$$
\end{lemma}

\begin{lemma}[\cite{Lengauer81}]
  For an undirected graph $G = (V,E)$, let $D_G = (V \cup E, E_G)$ be
  a DAG with $E_G = \{ (v, e) \,|\, e \in E, v \in V, v \in e\}$.
  Then $$\layout(G, V, \max) = \osbwp(D_G) + 2.$$
\end{lemma}

\subsection{Small Set Expansion Conjecture}\label{sec:defs:sse}

In this section we define the SSE Conjecture.  Let $G=(V,E)$ be an
undirected $d$-regular graph.  For a set $S \subseteq V$ of vertices,
we write $\Phi_G(S)$ for the (normalized) edge expansion of $S$,
$$
\Phi_G(S) = \frac{|E(S, V \setminus S)|}{d |S|}
$$
The Small Set Expansion Problem with parameters $\eta$ and $\delta$,
denoted $\SSE(\eta,\delta)$, asks if $G$ has a small set $S$ which
does not expand or whether all small sets are highly expanding.

\begin{defn}[$\SSE(\eta, \delta)$] \label{defn:sseprob}
Given a regular graph $G=(V,E)$, $\SSE(\eta, \delta)$ is the problem
of distinguishing between the following two cases:
\begin{description}
\item[Yes] There is an $S\subseteq V$ with $|S| = \delta |V|$ and $\Phi_G(S) \leq \eta$.
\item[No] For every $S \subseteq V$ with $|S| = \delta |V|$ it holds that $\Phi_G(S) \geq 1-\eta$.
\end{description}
\end{defn}

This problem was introduced by Raghavendra and Steurer \cite{RS10},
who conjectured that the problem is hard.

\begin{conj}[Small Set Expansion Conjecture]
  For every $\eta>0$, there is a $\delta>0$ such that $\SSE(\eta,\delta)$ is NP-hard.
\end{conj}

As has become common for a conjecture like this (such as the Unique
Games Conjecture), we say that a problem is \emph{SSE-hard} if it is
as hard to solve as the SSE problem.  Formally, a decision problem
$\mathcal{P}$ (e.g., a gap version of some optimization problem) is
\emph{SSE-hard} if there is some $\eta > 0$ such that for every
$\delta > 0$, $\SSE(\eta,\delta)$ polynomially reduces to $\mathcal{P}$.

Subsequently, Raghavendra et al.~\cite{RST10} showed that the SSE
Problem can in turn be reduced to a quantitatively stronger form of
itself.  To state this stronger version, we need to first define
Gaussian noise stability.

\begin{defn}
  Let $\rho \in [-1,1]$.  We define $\Gamma_\rho: [0,1] \rightarrow
  [0,1]$ by
  $$
  \Gamma_\rho(\mu) = \Pr\left[X \le \Phi^{-1}(\mu) \wedge Y \le \Phi^{-1}(\mu)\right]
  $$ where $X$ and $Y$ are jointly normal random variables with mean
  $0$ and covariance matrix $\left(\begin{array}{cc} 1 & \rho \\ \rho
  & 1
  \end{array}\right)$.
\end{defn}

The only fact we shall need about $\Gamma_\rho$ is the asymptotic
behaviour for $\rho$ close to $1$ and $\mu$ bounded away from $0$.

\begin{fact}
  \label{fact:gammabound}
  There is a constant $c > 0$ such that for all sufficiently small
  $\epsilon$ and all $\mu \in [1/10,1/2]$,
  $$
  \Gamma_{1-\epsilon}(\mu) \le \mu (1-c\sqrt{\epsilon}).
  $$
\end{fact}

We can now state the strong form of the SSE conjecture.

\begin{conj}[SSE Conjecture, Equivalent Formulation]
  \label{conj:main sse}
For every integer $q > 0$ and $\epsilon, \gamma >0$, it is NP-hard to distinguish between the following two cases for a given regular graph $G=(V,E)$
\begin{description}
\item[Yes] There is a partition of $V$ into $q$ equi-sized sets $S_1,
  \ldots, S_q$ such that $\Phi_G(S_i) \le 2\epsilon$ for
  every $1 \le i \le q$.
\item[No] For every $S\subseteq V$, letting $\mu = |S|/|V|$, it holds that $\Phi_G(S) \ge 1 -
  (\Gamma_{1-\epsilon/2}(\mu) + \gamma)/\mu$.
\end{description}
\end{conj}

For future reference, let us make two remarks about the strong form of
the conjecture.

\begin{remark}
  \label{rem:sse yes}
  In the \textbf{Yes} case of Conjecture~\ref{conj:main sse}, the
  number of edges leaving $S_i$ is at most
  $$|E(S_i, V \setminus S_i)| = \Phi_G(S_i) d |S| \le 4 \epsilon |E| /
  q.$$ In particular, the total number of edges that are not contained
  in one of the $S_i$'s is at most
  $$
  \frac{1}{2} \sum_{i} |E(S_i, V \setminus S_i)| \le 2 \epsilon |E|.
  $$
\end{remark}

\begin{remark}
  \label{rem:sse no}
  Using Fact~\ref{fact:gammabound} we see that, in the \textbf{No} case
  of Conjecture~\ref{conj:main sse}, we have
  $$
  \Phi_G(S) \ge c' \sqrt{\epsilon},
  $$ provided $\mu \in [1/10,1/2]$ and setting $\gamma \le
  \sqrt{\epsilon}$.  In particular, for every $|V|/10 \le |S| \le
  9|V|/10$, we have $|E(S, V \setminus S)| \ge c \sqrt{\epsilon} |E|$
  (switching roles of $S$ and $V \setminus S$ for $|S| > |V|/2$), for
  some universal constant $c$ (not the same constant as in
  Fact~\ref{fact:gammabound}).
\end{remark}

\section{Overview of Reductions}
\label{sec:overview}

We shall proceed as follows: for the two undirected edge problems
(i.e., MLA and MCLA), the hardness follows immediately from the strong
form of the SSE Conjecture (Conjecture~\ref{conj:main sse}) -- for the
case of MLA this was proved in \cite{RST10} and the proof for MCLA is
similar.  We then give in Section~\ref{sec:directed} a simple
reduction from MLA/MCLA to the four directed problems, and in
Section~\ref{sec:undir vertex} a similar reduction from MLA/MCLA to
the two undirected vertex problem.  The results for treewidth, which
are presented in the next section, follows from an additional analysis
of the instances produced by the reduction of Section~\ref{sec:undir
  vertex}.  Unfortunately, the results do not follow from hardness for
MLA/MCLA in a black-box way; for the soundness analyses we need to use
the expansion properties of the SSE instance.

We then give a reduction from MLA/MCLA with expansion, to the four
directed problems.  This reduction simply creates the bipartite graph
where the vertex set is the union of the edges and vertices of the
original graph $G$, with directed arcs from an edge $e$ to the
vertices incident upon $e$ in $G$.  The use of direction here is
crucial: it essentially ensures that both the vertex and edge counts
of any feasible ordering corresponds very closely to the number of
edges crossing the point in the induced ordering of $G$.  

To obtain hardness for the remaining two undirected problems, we
perform a similar reduction as for the directed case, creating the
bipartite graph of edge-vertex incidences.  However, since we are now
creating an undirected graph, we can no longer force the edges to be
chosen before the vertices upon which they are incident, which was a
key property in the reduction for the directed case.  In order to
overcome this, we duplicate each original vertex a large number of
times.  This gives huge penalties to orderings which do not
``essentially'' obey the desired direction of the edges, and makes the
reduction work out.

The results for treewidth, which are presented in
Section~\ref{sec:treewidth}, follows from an additional analysis of
the instances produced by the reduction for undirected vertex
problems.  Finally, the reduction for directed problems, implying hardness for
one-shot black pebbling, does not produce the kind of ``nice''
instances promised by Theorem~\ref{thm:osbp}.  In
Section~\ref{sec:pebbling postprocess}, we give some additional
transformation to achieve these properties.

\begin{figure}
  \centering \includegraphics[scale=0.55]{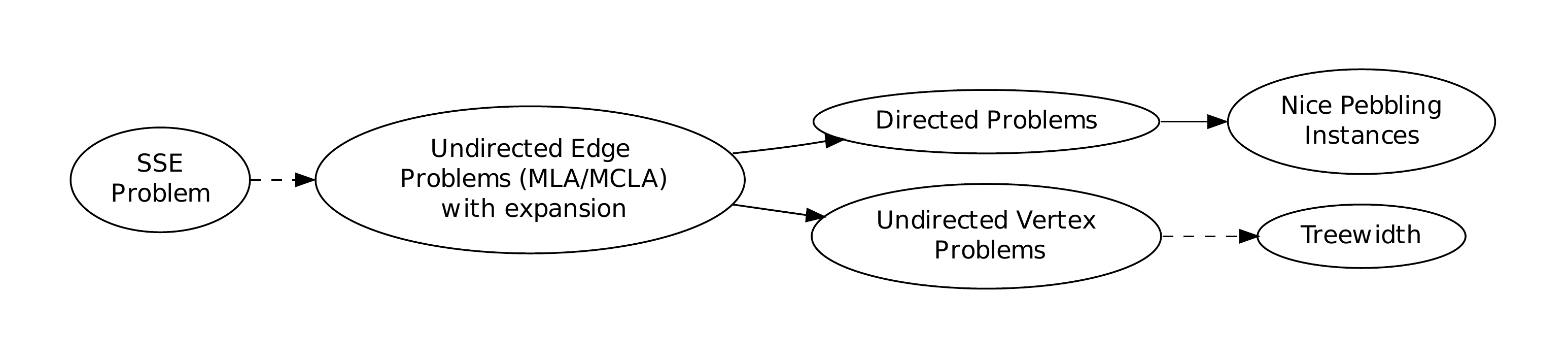}
  \caption{Overview of Reductions.  Dashed arrows indicate that the reduction is 
obtained by the identity mapping, whereas solid arrows indicate a nontrivial transformation
 from one problem to the other.}
  \label{fig:reductions-overview}
\end{figure}

Figure~\ref{fig:reductions-overview} gives a high-level overview of
these reductions.

\section{Hardness For Layout Problems}
\label{sec:layout}

In this section, we show that all of the layout problems defined in
Section~\ref{sec:layout-defs} are SSE-hard to approximate within
any constant factor.  This also shows that pathwidth and the one-shot
pebbling problems are hard to approximate within any constant.

\subsection{Hardness for MCLA and MLA} \label{sec:mcla-proof}

In this section, we recall the proof of \cite{RST10} for MLA, and
observe that it applies for MCLA as well.  For an undirected graph
$G$, let us write $\MCLA(G)$ (resp., $\MLA(G)$) for the MCLA value
(resp., MLA value) of $G$, i.e.,
\begin{align*}
\MLA(G) &= \layout(G; E, \Sigma) = \min_{\pi} \sum_{i \in [n]} |E_i(\pi)| \\
\MCLA(G) &= \layout(G; E, \max) = \min_{\pi} \max_{i \in [n]} |E_i(\pi)|.
\end{align*}

\begin{theorem}
  \label{thm:mla-mcla}
  For every $\epsilon > 0$, given a graph $G = (V,E)$, it is SSE-hard
  to distinguish between:
  \begin{description}
  \item[Yes] $\MLA(G) \le O(\epsilon \cdot |V| \cdot |E|) $ and $\MCLA(G) \le O(\epsilon |E|)$
  \item[No] For every $S \subseteq V$ with $|V|/10 \le |S| \le
    9|V|/10$, it holds that $|E(S, V \setminus S)| \ge
    \Omega(\sqrt{\epsilon} |E|)$.  In particular, $\MLA(G) \ge
    \Omega(\sqrt{\epsilon} \cdot |V| \cdot |E|)$ and $\MCLA(G) \ge
    \Omega(\sqrt{\epsilon} |E|)$.
  \end{description}
\end{theorem}

\begin{proof}

We use the instances for Conjecture \ref{conj:main sse} with $q =
1/\epsilon$.  Let $G = (V,E)$ be an instance for Conjecture~\ref{conj:main sse}.

In the \textbf{Yes} case, we have disjoint sets $S_1,\ldots,S_q$ and
for each set $S_j$, $|S_j| = n/q= \epsilon n$, $\Phi_G(S_j) \leq
2\epsilon$.  We give an ordering $\pi$ of the vertices such that
$\max_{i \in [n]} |E_i(\pi)| \leq 3\epsilon |E|$ as follows.  Order
the vertices as $S_1, \ldots, S_q$ (with the order within each $S_j$
chosen arbitrary) and let this order be $\pi$.  For any $i \in [n]$,
we show that $|E_i(\pi)| \leq 3\epsilon |E|$.  Suppose $\pi^{-1}(i)$ is a
vertex in $S_j$.  Each edge in $E_i(\pi)$ either has both end-points
inside $S_j$, or its end-points in two different $S_k$'s.  The total
number of edges inside $S_j$ is at most $\epsilon d n / 2 = \epsilon
|E|$.  Moreover, by Remark~\ref{rem:sse yes}, the total number of
edges with end-points in two different $S_k$'s is at most $2 \epsilon
|E|$.  Therefore, $\MCLA(G) \le \max_{i} |E_i(\pi)| \le 3 \epsilon
|E_i(\pi)|$.  The $\MLA$ value can be bounded similarly.

The property of the \textbf{No} instance is the same as in
Conjecture~\ref{conj:main sse} (via Remark~\ref{rem:sse no}), and the
implications for the $\MLA$ and $\MCLA$ values are immediate.
\end{proof}

\subsection{Reduction To Directed Graphs}
\label{sec:directed}

Given an undirected graph $G = (V,E)$, we construct a directed graph
$D = (V', E')$ as follows.  In order to distinguish the elements of
$V$ and $E$ from the elements of $V'$ and $E'$, we refer to elements
of $V$ as vertices, elements of $E$ as edges, elements of $V'$ as
nodes, and elements of $E'$ as arcs.

There is a node in $D$ for each vertex and for each edge of $G$,
i.e., $V' = V \cup E$.  The graph $D$ is bipartite with bipartition
$V, E$, and there is an arc in $D$ from $e \in E$ to $v \in V$
if $e$ is incident upon $v$.  Formally,
\begin{eqnarray*}
 V' &=& V \cup E \\
 E' &=& \{ (e, v) \,|\,e \in E, v \in V, v \in e\}.
\end{eqnarray*}
See also Figure~\ref{fig:reduction1}.

\begin{figure}[!t]
\begin{center}
\input{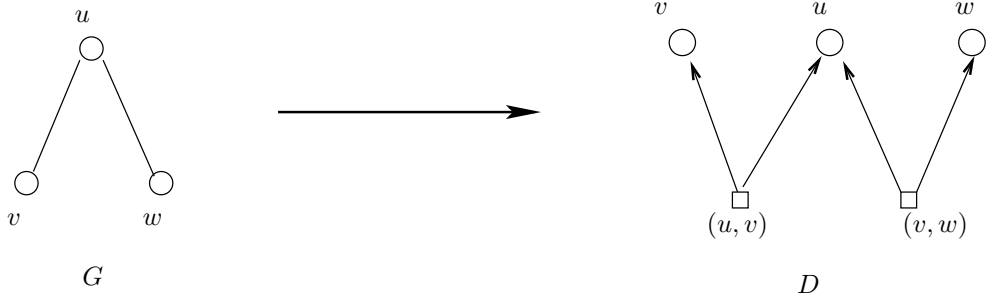}
\caption{The reduction from $G$ to $D$.}
\label{fig:reduction1}
\end{center}
\end{figure}

The remainder of this section is devoted to analyzing the reduction.
First, it is easy to give an upper bound on the four $\layout$ values
of $D$ in terms of the $\MLA$ and $\MCLA$ values of $G$.

\begin{lemma}
  \label{lem:DAG-completeness}
  The DAG $D$ constructed from $G$ as above satisfies the following:
  \begin{align*}
    \layout(D; E, \Sigma) & \le \left( \MLA(G) + O(|E|) \right) \cdot (d+1) \\
    \layout(D; E, \max) & \le \MCLA(G) + d.
  \end{align*}
\end{lemma}

Note that, for the purposes of applying this to the graphs of
Theorem~\ref{thm:mla-mcla} the error term of $|E|$ (resp.\ $d$) is
insignifcant compared to the $\MLA$ (resp.\ $\MCLA$) value of $G$.

\begin{proof}
  Consider an ordering $\pi$ of $V$.  For a set
  of vertices $S$ of $V$, let $u_{\pi}(S) \in S$ denote the vertex of $S$
  that comes first in the ordering $\pi$.  

  We extend $\pi$ to an ordering $\pi'$ of $V'$ by inserting each edge
  $e = (u,v)$ immediately before the vertex $u_\pi(e)$.  It is easy to
  see that for each node $z \in V' \cup E'$,
  \begin{align*}
    |E_{z}(\pi')| & \le |E_{u_{\pi}(z)}(\pi)| + d 
  \end{align*}
  This immediately implies
  \begin{equation*}
    \layout(D; E, \max) \le \max_{z \in V'} |E_{z}(\pi')| \le \max_{u \in V} |E_{u}(\pi)| + d,
  \end{equation*}
  Setting $\pi$ to be an
  optimal MCLA ordering of $G$, we obtain the second claim of the
  Lemma.  Similarly, using that $|u^{-1}(v)| \le d+1$ for every $v \in V$, we get
  \begin{equation*}
    \layout(D; E, \Sigma) \le \sum_{z \in V'} |E_{z}(\pi')| \le (d+1) \sum_{u \in V} |E_{u}(\pi)| + d |V'|,
  \end{equation*}
  Setting $\pi$ to be
  an optimal MLA ordering of $G$ and using $|V'| = O(|E|)$, we obtain
  the first claim of the Lemma.
\end{proof}

Next we use the strong soundness property of Theorem~\ref{thm:mla-mcla}
to argue about the soundness of $D$.

\begin{lemma}
  \label{lem:DAG-soundness}
  Suppose $G$ has the property that for every $|V|/10 \le |S| \le
  9|V|/10$ we have $|E(S, V \setminus S)| \ge \Omega(\sqrt{\epsilon}
  |E|)$.  Then,
  \begin{align*} 
    \layout(D; V, \Sigma) & \ge \Omega(\sqrt{\epsilon} |E|^2) \\
    \layout(D; V, \max) & \ge \Omega(\sqrt{\epsilon} |E|) 
  \end{align*}
\end{lemma}

\begin{proof}
  Let $\pi'$ be any ordering of $V'$.  Using the expansion property of
  Theorem~\ref{thm:mla-mcla}, we'll show that this ordering must have
  high cost.  For a point $i \in [N]$, let $S_i$ be the set of
  vertices of $V$ that appear \emph{after} $i$ in $\pi'$.

  The bound on $\layout(D; V, \max)$ is immediate: consider a point $i
  \in [N]$ such that $|S_i| = |V|/2$.  By the expansion property
  $|E(S_i, V \setminus S_i)| \ge \Omega(\sqrt{\epsilon}|E|)$, and
  since each such edge $e$ has one of its endpoints before point $i$,
  the node $e$ itself must appear before point $i$ and thus
  $|V_i(\pi')| \ge |E(S_i, V \setminus S_i)| \ge
  \Omega(\sqrt{\epsilon} |E|)$.

  Let us then turn to $\layout(D; V, \Sigma)$.  Write $c_i$ for the
  fraction of edges $e$ that appear \emph{before} (or at) point $i$ in
  $\pi'$.  We shall show that whenever $1/5 \le c_i \le 4/5$, we
  have $|V_i(\pi')| \ge \Omega(\sqrt{\epsilon}|E|)$, giving a total of
  $\layout(D; V, \Sigma) \ge \Omega(\sqrt{\epsilon}|E|^2)$.

  By a simple counting argument, we have
  $$
  d |S_i| \ge 2 (1-c_i) |E|,
  $$ implying $|S_i| \ge (1-c_i)|V|$ which for $c_i \le 4/5$ is at
  least a $1/5$ fraction of vertices.  If in addition $|S_i| \le
  9|V|/10$, the argument above gives $|V_i(\pi')| \ge
  \Omega(\sqrt{\epsilon} |E|)$.  The remaining case is that $|S_i| \ge
  9 |V|/10$.  But then $S_i$ is incident upon at least a $9/10$
  fraction of edges.  This implies that the number of edges incident
  upon $S_i$, appearing before $i$ in $\pi'$, are at least
  $|E|(c_i-1/10)$ which for $c_i \ge 1/5$ is $\Omega(\sqrt{\epsilon}
  |E|)$.
\end{proof}

Combining Lemma~\ref{lem:DAG-completeness} and
Lemma~\ref{lem:DAG-soundness}, with Theorem~\ref{thm:mla-mcla}, and
using the fact that edge costs are always larger than the
corresponding vertex costs, we immediately obtain the following
theorem.

\begin{theorem}
  \label{thm:directed}
  Given a DAG $D$, $\layout(D; E, \max)$, $\layout(D; E, \Sigma)$,
  $\layout(D; V, \max)$, and $\layout(D; V, \Sigma)$ are all SSE-hard
  to approximate within any constant factor, even in DAG's with
  maximum path length $1$ (i.e., every vertex is a source or a sink).
\end{theorem}

\begin{remark}
  In fact we see that, as in Theorem~\ref{thm:mla-mcla}, the four
  hardness results applies to the same instance, so that it is SSE-hard
  to distinguish all of the four $\layout$ values being high from all of
  them being low.  
\end{remark}

As the one-shot black pebbling problem is precisely $\layout(D; V,
\max)$, we obtain hardness for one-shot black pebbling as an immediate
corollary.  However, the instances are not single-sink DAGs with
maximum indegree $2$, as promised in Theorem~\ref{thm:osbp}.  In
Section~\ref{sec:pebbling postprocess} we show how to transform the
instances further to obtain such DAGs.

\subsection{Undirected Vertex Problems}
\label{sec:undir vertex}

The reduction for undirected vertex problems is very similar to the
reduction for directed problems given in the previous section.  As
before, we introduce nodes for every edge of $G$.  As in the directed
case, we are interested in orderings where an edge appears before its
two endpoints, but we cannot use direction to force this anymore.
Instead, we ensure that orderings that are not like this incur a
high cost by replicating each node corresponding to a vertex of $G$
many times.

Given an undirected graph $G = (V,E)$, we construct a new graph
$G' = (V', E')$ as follows.

There are $r$ nodes in $G'$ for each vertex and one node for each edge
of $G$, i.e., $V' = V \times [r] \cup E$.  For a vertex $u \in V$ we
write $u^{1}, \ldots, u^{r}$ to denote the $r$ copies of $u$ and refer
to each such set of $r$ nodes as a \emph{vertex group}.  The graph
$G'$ is bipartite with bipartition $V \times [r], E$, and there is an
edge in $G'$ between $e \in E$ and $v^{i} \in V \times [r]$ if $e$ is
incident upon $v$.  Formally,
\begin{eqnarray*}
  V' &=& \{ v^i \,|\, v \in V, i \in [r]\}\\
  E' &=& \{ (e, v^i) \,|\,e \in E, v \in V, v \in e, i \in [r]\}.
\end{eqnarray*}
See also Figure~\ref{fig:reduction2}.

\begin{figure}[t]
\begin{center}
\input{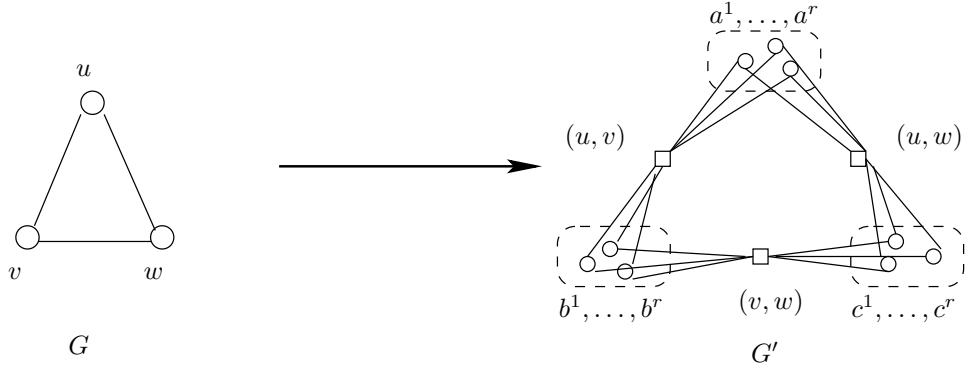}
\caption{The reduction from $G$ to $G'$, illustrated for $r = 3$.}
\label{fig:reduction2}
\end{center}
\end{figure}

\begin{lemma}
  \label{lem:undir-completeness}
  The graph $G'$ constructed from $G$ as above satisfies the following:
  \begin{align*}
    \layout(G'; V, \Sigma) & \le (d+r) \MLA(G) \\
    \layout(G'; V, \max) & \le \MCLA(G).
  \end{align*}
\end{lemma}

\begin{proof}
  We proceed as in the proof of Lemma~\ref{lem:DAG-completeness}.  An
  ordering $\pi$ of $V$ naturally induces an ordering $\pi'$ of $V'$:
  put all $r$ copies of $u \in V$ consecutively, with vertices of $V$
  appearing in the same order as in $\pi$, and insert each edge $e \in
  E$ immediately before its first vertex.  Again, for an edge $e \in
  E$, let $u_{\pi}(z)$ denote the endpoint of $e$ that appears first
  in $\pi$.  Similarly, for a copy $v^i \in V'$ of $v \in V$, let
  $u_{\pi}(v^i) = v$.  It is easy to see that the constructed ordering
  $\pi'$ satisfies
  \begin{align*}
    |V_{z}(\pi')| & \le |E_{u_{\pi}(z)}(\pi)|
  \end{align*}
  for every $z \in V'$.  This immediately implies
  \begin{equation*}
    \layout(G'; V, \max) \le \max_{z \in V'} |V_{z}(\pi')| \le \max_{u \in V} |E_{u}(\pi)|,
  \end{equation*}
  Similarly, we get
  \begin{equation*}
    \layout(G'; V, \Sigma) \le \sum_{z \in V'} |V_{z}(\pi')| \le (d+r) \sum_{u \in V} |E_{u}(\pi)|.
  \end{equation*}
\end{proof}

\begin{lemma}
  \label{lem:undir-soundness}
  Suppose $G$ has the property that for every $|V|/10 \le |S| \le
  9|V|/10$ we have $|E(S, V \setminus S)| \ge \Omega(\sqrt{\epsilon}
  |E|)$.  Then, if $r \ge |V| \cdot |E|$, we have
  \begin{align*} 
    \layout(G'; V, \Sigma) & \ge \Omega(\sqrt{\epsilon} \cdot r \cdot |V| \cdot |E|) \\
    \layout(G'; V, \max) & \ge \Omega(\sqrt{\epsilon} |E|) 
  \end{align*}
\end{lemma}

\begin{proof}
  Let $\pi'$ be an ordering of $V'$.  First we have the following
  simple claim, establishing that for good orderings, most vertices appear
  after their edges.
  \begin{claim}
    Suppose that for some vertex $u \in V$, at least $r/2$ of the
    copies of $u$ in $G'$ appear before some edge $e = (u,v) \in E$
    adjacent upon $u$.  Then 
    \begin{align*}
      \max_{i \in [N]} |V_i(\pi')| & \ge r/4 \gg \Omega(\sqrt{\epsilon} \cdot |E|) \\
      \sum_{i \in [N]} |V_i(\pi')| & \ge (r/4)^2 \gg \Omega(\sqrt{\epsilon} \cdot r \cdot |V| \cdot  |E|).
    \end{align*}
  \end{claim}
  
  \begin{proof}
    Let $I_1$ be the first half of the positions where copies of $u$
    appear before $e$, and $I_2$ the second half.  Thus, $|I_1|, |I_2|
    \ge r/4$.  Then each element of $I_1$ contributes to $V_i(\pi')$
    for each $i \in I_2$, giving the claimed bounds.
  \end{proof}

  Thus we may without loss of generality assume that for each vertex
  $u$ of $V$, at least $r/2$ of its $r$ copies in $G'$ appear
  \emph{after} all edges adjacent upon $u$.  From now on, let us
  discard all the $\le r/2$ ``bad'' copies of each vertex of $V$ that
  appear before some of its edges.  This only decreases the cost of
  $\pi'$, and there are still $\ge r|V|/2$ vertex nodes left.

  Let $i_1$ be the (first) point of $\pi'$ such that $r|V|/10$ vertex
  nodes are to the left of $i_1$, and $i_2$ the (last) point of $\pi'$
  such that $r|V|/10$ vertex nodes are to the right of $i_2$.  

  \begin{claim}
    \label{claim:undir-soundness}
    For any point $i$ between $i_1$ and $i_2$, we have $|V_i(\pi')| \ge
    \Omega(\sqrt{\epsilon} |E|)$.
  \end{claim}

  \begin{proof}
    Let $S \subseteq V$ (resp.\ $T \subseteq V$) be the set of
    vertices $u$ such that some copy of $u$ appears before $i$
    (resp.\ after $i$).  We then have $|S|, |T| \ge
    \frac{r|V|/10}{r/2} \ge |V|/5$, and $S \cup T = V$.  Thus we can
    partition $V$ into $S' \subseteq S, T' \subseteq T$ such that
    $|S'|, |T'| \le 4|V|/5$.  By the expansion property of $G$ we have
    $|E(S', T')| \ge \Omega(\sqrt{\epsilon}|E|)$.  Further, we also
    have $|V_i(\pi')| \ge |E(S', T')|$ as each $e \in E(S',T')$ must
    appear before $i$ in $\pi'$ (because one of their endpoints is in
    $S$) but have an edge crossing $i$ (because the other of their
    endpoints is in $T$).
  \end{proof}

  From Claim~\ref{claim:undir-soundness}, the proof of the lemma
  follows immediately.
\end{proof}

As in the previous section, we can now combine
Lemma~\ref{lem:undir-completeness} and
Lemma~\ref{lem:undir-soundness}, with Theorem~\ref{thm:mla-mcla}, to obtain:

\begin{theorem}
  \label{thm:undir-vertex}
  Given a graph $G$, $\layout(G; V, \max)$, $\layout(G; V, \Sigma)$
  are both SSE-hard to approximate within any constant factor, even in
  bipartite graphs.
\end{theorem}

As the pathwidth problem is precisely $\layout(G; V, \max)$, we obtain
hardness for pathwidth as an immediate corollary.  In
the next section, we'll show the stronger soundness
required for Theorem~\ref{thm:treewidth}.

\section{Hardness For Treewidth}\label{sec:treewidth}

In this section we shall complete our proof of
Theorem~\ref{thm:treewidth} by showing that the hard instances for
pathwidth from Theorem~\ref{thm:undir-vertex} also have large
treewidth.

\begin{lemma}
  \label{lem:treewidth-soundness}
  Let $G = (V,E)$ be an undirected graph with the property that for
  every $|V|/10 \le |S| \le 9|V|/10$ we have $|E(S, V \setminus S)|
  \ge \Omega(\sqrt{\epsilon} |E|)$, and let $G'$ be the graph obtained
  by applying the reduction of Section~\ref{sec:undir vertex} to $G$.
  Then, if $r \ge |V| \cdot |E|$,
  we have
  \begin{align*} 
    \tw(G') & \ge \Omega(\sqrt{\epsilon} |E|)
  \end{align*}
\end{lemma}

To prove Lemma~\ref{lem:treewidth-soundness}, we shall use the fact
that the treewidth of a graph is closely related to an expansion-like
property called the $1/2$-separator number, defined in \cite{BGHK95}.

\begin{defn}[1/2-vertex separator, 1/2-separator number]
  Let $G=(V,E)$ be an undirected graph. For $W \subseteq V$, a
  \emph{$1/2$-vertex separator} of $W$ in $G$ is a set $S\subseteq V$ of
  vertices such that every connected component of the graph $G[V-S]$
  contains at most $|W|/2$ vertices of $W$.  Let $\psi_G(1/2, W)$
  denote the minimum size of a 1/2-vertex separator of $W$ in
  $G$.  We define the \emph{1/2-separator number} $K_{1/2}(G)$ to be
  $$K_{1/2}(G) = \max_{W\subseteq V} \psi_G(1/2,W).$$
\end{defn}

\begin{lemma}[\cite{BGHK95}]  \label{lemma:separator}
  For every graph $G = (V,E)$, it holds that $\tw(G) \ge K_{1/2}(G)-1$.
\end{lemma}

Using this, it is now straightforward to prove the lower bound on the
treewidth.

\begin{proof}[Proof of Lemma~\ref{lem:treewidth-soundness}]
    We'll show that $\psi_{G'}(1/2, V') \ge \Omega(\sqrt{\epsilon}
    |E|)$ (i.e. we choose $W=V'$).  Suppose $C$ is an optimal
    $1/2$-vertex separator of $V'$ and it separates $V'\setminus C$
    into $l$ sets $V_1', \ldots, V_l'$, each of size at most $|V'|/2$.
    By merging different $V_i'$ we may assume that we only have two
    sets $V_1'$ and $V_2'$, both of size \emph{at least} $|V'|/5$.

    Now, similarly to the proof of Claim~\ref{claim:undir-soundness}, let
    $S \subseteq V$ (resp.\ $T \subseteq V$) be the set of vertices
    $v$ such that some copy of $V$ appears in $V_1'$ (resp.\ $V_2'$).
    As $|V_1'|, |V_2'| \ge r |V|/5$, this implies that both $|S|, |T|$
    are at least $|V|/5$, and furthermore $S \cup T = V$ (since
    otherwise all $r$ copies of some vertex are in $C$, implying $|C|
    \ge r \gg \Omega(\sqrt{\epsilon}|E|)$).  We can thus choose a
    balanced partition $S', T'$ such that $S' \subseteq S$, $T'
    \subseteq T$, and we have $|E(S',T')| \ge
    \Omega(\sqrt{\epsilon}|E|)$.  But every edge $e = (u,v)$ such that
    $u \in S'$ and $v \in T'$ must belong to $C$, since it is
    connected (in $G'$) to every copy of $u$ and $v$.
\end{proof}

\section{Nicer Pebbling Instances}
\label{sec:pebbling postprocess}

In this section we show how to transform our hard instances for
one-shot black pebbling so as to have in-degree bounded by $2$ and
single sinks.

We begin with the in-degree.

\begin{lemma}
  \label{lemma:indeg2}
  Given a DAG $D = (V,E)$ we can in polynomial time construct a DAG
  $D' = (V', E')$ such that every node of $D'$ has in-degree at most
  $2$ and
  $$
  \osbp(D) \le \osbp(D') \le \osbp(D)+d,
  $$
  where $d$ is the maximum in-degree of $D$.
\end{lemma}

As the proof of this lemma is somewhat lengthy, we defer it until
later in this section and first describe how to obtain a DAG with a
single sink.

\begin{lemma}
  \label{lemma:onesink}
  Given a DAG $D = (V,E)$ we can in polynomial time construct a DAG
  $D' = (V', E')$ such that $D'$ has a single sink and
  $$
  \osbp(D) \le \osbp(D') \le \osbp(D) + s+1,
  $$ where $s$ is the number of sinks in $D$.  Furthermore, if $D$ has
  maximum in-degree $2$ then so does $D'$.
\end{lemma}

\begin{proof}
  Construct $D'$ by adding a binary tree with $s$ leaves to $D$, and
  identifying the leaves of the tree with the sinks of $D$.  The
  properties of $D'$ are easily verified.  Since $D'$ is a super-DAG
  of $D$, its pebbling cost must be at least $\osbp(D)$.  Conversely,
  a valid pebbling of $D'$ can be obtained by using a one-shot
  pebbling of $D$ but without removing pebbles from the sinks, and
  then pebbling the tree.
\end{proof}

Given Lemmas~\ref{lemma:indeg2} and \ref{lemma:onesink} we can derive
Theorem~\ref{thm:osbp}.  The point is that for the graphs produced in
Theorem~\ref{thm:directed}, the maximum indegree $d$ equals the degree
of the original SSE instance, and the number of sinks $s$ equals the
number of vertices of the original pebbling instance.  On the other
hand, the one-shot black pebbling cost (i.e., the $\layout(D; V,
\max)$ value) is of order $\epsilon |E|$ which is much larger than $d$
and $s$, so the additive losses of Lemmas~\ref{lemma:indeg2} and
\ref{lemma:onesink} are insignificant.  We omit the details.

\subsection{Lemma~\ref{lemma:indeg2}}

  In this section we prove Lemma~\ref{lemma:indeg2}.  First, recall
  the definition of a pyramid graph.

  \begin{figure}[t]
    \centering
    \includegraphics{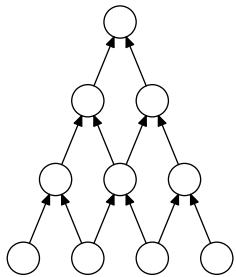}
    \caption{Pyramid of size $4$}
    \label{fig:pyramid}
  \end{figure}

  \begin{defn}
    A \emph{pyramid graph} of size $d$ is a layered graph of indegree
    two, with $d$ layers, labelled $0, 1, \ldots, d-1$.  Layer zero
    (the input layer) consists of $d$ vertices, and layer $i$ contains
    $d-i$ vertices.  See Figure~\ref{fig:pyramid}.
  \end{defn}

  \begin{figure}[t]
    \centering
    \includegraphics{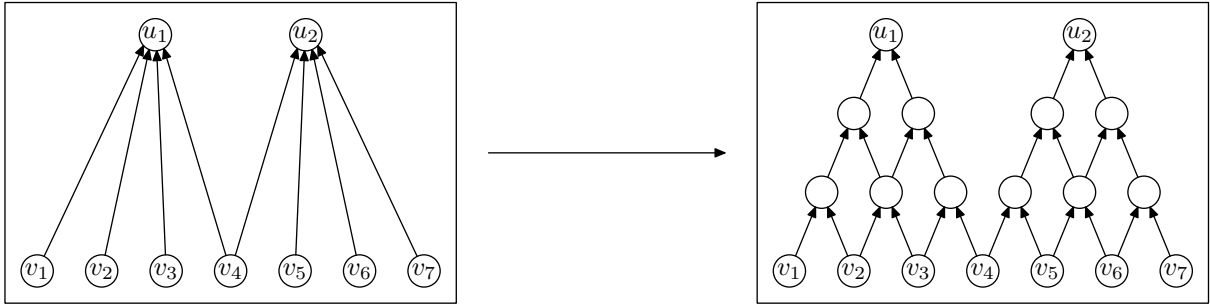}
    \caption{Reduction to DAGs of indegree $2$}
    \label{fig:indeg2reduction}
  \end{figure}

  The reduction of Lemma~\ref{lemma:indeg2} to produce DAGs of
  indegree $2$ is as follows.  Construct $D'$ by replacing each vertex
  $u$ by a pyramid $P_u$ of size $d(u)$ (here, $d(u)$ denotes the
  indegree of $u$), where the $d(u)$ vertices at layer $0$ of $P_u$
  are identified with the predecessors of $u$, and $u$ is identified
  with the vertex at layer $d(u)-1$ of $P_u$.  See
  Figure~\ref{fig:indeg2reduction}.

  To prove the lemma we need to show that $D'$ constructed this way
  satisfies $$\osbp(D) \le \osbp(D') \le \osbp(D') + d,$$ where $d$ is
  the maximum indegree of any vertex $u$ of $D$.

  In what follows whenever we say ``pebbling strategy'' of $D$ or $D'$
  we always refer to a one-shot black pebbling strategy of $D$ or
  $D'$.
  
  The upper bound on $\osbp(D')$ is trivial: if $S$ is a valid
  pebbling strategy for $D$, then clearly we can create
  a corresponding pebbling strategy for $D'$ by
  pebbling through the pyramid whenever $D$ pebbles the sink of the
  pyramid.  This takes at most $d$ additional pebbles.  

  In the other direction, we want to show that a pebbling strategy for
  $D'$ can be converted into a pebbling strategy for $D$.  We will
  first show that $D'$ can be assumed to be in a particular normal
  form, and then using this normal form, we will show how to simulate
  the pebbling.

  \begin{defn}
    Let $S'$ be a pebbling strategy of $D'$. That is, $S'$ is a sequence of
    configurations, where each configuration is a set of black pebble
    placements, and such that the sequence of configurations follows
    by the black pebbling rules. We say that configuration $c \in S'$ is
    saturated with respect to a pyramid $P_u$ if $c$ is the first time in
    $S'$ that there is a black pebble path cutting the sink of $P_u$ from
    all of the sources of $P_u$. (The cut does not include any sources
    or the sink of $P_u$.)  Note that this cut has size $d-1$.
  \end{defn}

  \begin{claim} 
    Let $S'$ be a pebbling strategy for $D'$. We can assume without
    loss of generality that $D'$ has the following normal form.  For
    each configuration $c' \in S'$, if $c'$ is saturated with respect
    to pyramid $P_u$, then the subsequent moves of $S'$ pebble the sink
    of $P_u$ (in the obvious way), removing all other black pebbles on
    the internal nodes of $P_u$.
  \end{claim}

  \begin{proof}[Proof Sketch]
    At a saturated configuration $c'$, there must be $d-1$ pebbles on
    internal nodes of $P_u$. If we subsequently pebble the sink of
    $p$, we will never use more than $d-1$ pebbles on internal nodes
    of $p$, and all other pebbles on the graph stay as they were. Thus
    the normal form does not use more pebbles than the original
    strategy.  Furthermore, since the internal nodes of a pyramid are
    only used to pebble the sink of this pyramid, we have not lost
    anything by pebbling through to the sink and removing the other
    internal black pebbles.
  \end{proof}

  From now on we will assume that the pebbling $S'$ of $D'$ has the
  above normal form. That is, if a configuration is saturated (with
  respect to a pyramid $P_u$), the next thing that happens in $S'$ is
  to pebble the sink of $P_u$. (After pebbling the sink, we will have
  not touched whatever pebbles were on the source nodes of $P_u$, and
  we will have a pebble on the sink node of $P_u$, and no other
  internal pebbles on $P_u$.)
 
  Our strategy for constructing a pebbling, $S$, of $D$, given a 
  normal form pebbling, $S'$ of $D'$ is as follows.
  For each node $v$ of $D$, pebble $v$ whenever it is first pebbled in $S'$,
  and remove the pebble from $v$ as soon as all successors
  of $v$ (in the original graph $D$) are pebbled.
  We want to argue that this pebbling strategy of $D$ is not greater
  than that of $D'$. To see this, we will use the following Lemma.

  \begin{lemma}
    \label{lemma:frugal}
    In any frugal read-once black strategy of a size $d$ pyramid, the
    number of pebbles on the pyramid at any point in time, up until
    all sources are pebbled for the first time, must be equal to the
    number of sources in that pyramid that have been pebbled so far.
  \end{lemma}

  Assuming the above Lemma it is clear that if $S'$ is a normal form
  pebbling of $D'$, then for any pyramid $p$ in $D'$, and any
  configuration $c'$, if there are $k$ pebbles on $P_u$ at $c'$, then
  in the corresponding configuration $c$ of $D$, there are at most $k$
  pebbles on source nodes of $P_u$.  To see this, first notice that by
  the above Lemma, anytime a pyramid is being pebbled in $D'$ up until
  the time when all source nodes of the pyramid are pebbled for the
  first time, the number of pebbles on the pyramid will be at least as
  large as the number of source nodes in $D$ that contain pebbles.
  Then by the normal form property of $D'$, as soon as all source
  nodes of $D'$ are pebbled for the first time, the strategy pebbles
  the sink of $D'$, and thus again the number of corresponding pebbles
  on $D$ is never greater than the number of pebbles on $D'$.

  \begin{proof}[Proof of Lemma~\ref{lemma:frugal}]
    Let $P$ be a size $d$ pyramid graph, and let $S$ a one-shot black
    pebbling of $P$. Let $c$ be a configuration occuring in $S$ such
    that the set of source nodes that have been pebbled up to $c$ are
    the source nodes of $P'$, where $P'$ is a size $d'$ sub-pyramid
    of $P$.  We want to argue that $c$ must contain at least $d'$
    pebbles.  Assume without loss of generality that $P'$ is the
    leftmost sub-pyramid of $P$, of size $d' < d$.  Label the outer
    rightmost vertices of $P'$ by $v_{d-1}, ..., v_0$, where $v_{d-1}$ is
    the sink vertex of $P'$, and for all $i < d-1$, $v_i$ is the rightmost
    vertex in $P'$ at level $i$. 
    Corresponding to each named vertex $v_i$ is a
    diagonal set of vertices, $diag(v_i)$, beginning at $v_i$ and
    travelling southwest to a source vertex of $P'$.  Note that
    the sets $diag(v_i)$ are pairwise disjoint. We will argue
    that for each $i$, $0 \leq i \leq d-1$, at least one vertex from
    $diag(v_i)$ must appear in $c$.  To see this, first notice that
    for each $v_i$, there is a vertex $v_i'$ that is an immediate
    successor of $v_i$ and that lies outside of $P'$. This vertex
    $v_i'$ must be pebbled at some time after configuration $c$, since
    it has a predecessor that has not yet been pebbled.  But in order
    to pebble $v_i'$ in the future, there must be a black pebble on
    some vertex in $diag(v_i)$ in $c$.  Thus, we have shown that if
    $c$ is any configuration in $S$ such that $d' <d$ source vertices
    are pebbled thus far, then there must be $d'$ vertices pebbled in
    $c$.
\end{proof}

\section{Conclusion and Open Problems}  \label{sec:conclusion}

We proved SSE-hardness of approximation for a variety of graph
problems.  Most importantly we obtained the first inapproximability
result for the treewidth problem.  

Some remarks are in order.  The status of the SSE conjecture is, at
this point in time, very uncertain, and our results should therefore
not be taken as absolute evidence that there is no polynomial time
approximation for (e.g.) treewidth.  However, at the very least, our
results do give an indication of the difficulty involved in obtaining
such an algorithm for treewidth, and builds a connection between these
two important problems.  We also find it remarkable how simple our
reductions and proofs are.  We leave the choice of whether to view
this as a healthy sign of strength of the SSE Conjecture, or whether
to view it as an indication that the conjecture is too strong, to the
reader.

There are many important open questions and natural avenues for
further work, including:

\begin{enumerate}
\item
  It seems plausible that these results can be extended to a wider
  range of graph layout problems.  For instance, our two choices of
  aggregators $\max$ and $\Sigma$ can be viewed as taking $\ell_{\infty}$
  and $\ell_{1}$ norms, and it seems likely that the results would apply
  for any $\ell_{p}$ norm (though we are not aware of any previous
  literature studying such variants).

\item 
  It would be nice to obtain hardness of approximation result for our
  problems based on a weaker hardness assumption such as UGC. It is
  conjectured in \cite{RST10} that the SSE conjecture is equivalent to
  UGC.   Alternatively, it would be nice to show that hardness of some of our
  problems imply hardness for the SSE Problem.

\item
  For pebbling, it would be very interesting to obtain results for the
  unrestricted pebbling problems (for which finding the exact pebbling
  cost is even PSPACE-hard).  As far as we are aware, nothing is known
  for these problems, not even, say, whether one can obtain a
  non-trivial approximation in NP.  As mentioned in the introduction,
  we are currently working on extending our one-shot pebbling results
  to bounded time pebblings.  We have some preliminary progress there
  and are hopeful that we can relax the pebbling results to a much
  larger class of pebblings.

\end{enumerate}

%\section{Acknowledgements}

\bibliographystyle{alpha}
\bibliography{references}

\end{document}